\documentclass[runningheads]{llncs}
\usepackage[T1]{fontenc}
\usepackage{graphicx}
\usepackage{xspace}

\usepackage{listings}
\usepackage{graphicx}

\usepackage{todonotes}[disabled]
\usepackage[normalem]{ulem}
\usepackage{xifthen}
\usepackage{amsmath}
\usepackage{amsfonts}
\usepackage{amssymb}
\usepackage{caption}
\usepackage{subcaption}
\usepackage{multirow}
\usepackage{hyperref}
\usepackage{paralist}
\newcommand{\multiset}[1]{\boldsymbol{#1}}

\newcommand{\os}{\ensuremath{\mathfrak{E}}}
\newcommand{\tup}[1]{\langle#1\rangle}
\newcommand{\id}{\mathit{id}}
\newcommand{\supp}{\mathit{supp}}

\newcommand{\sqleq}{\sqsubseteq}

\newcommand{\support}[1]{\texttt{Supp}(#1)}

\newcommand{\prefun}{{\ensuremath{\tt pre}}}
\newcommand{\postfun}{{\ensuremath{\tt post}}}

\newcommand{\out}{out_{\Upsilon}(t)}
\newcommand{\abs}[1]{\lvert#1\rvert}

\usepackage{tikz}
\usetikzlibrary{patterns,shapes,arrows,calc,fit,arrows.meta}
\tikzstyle{place}=[circle,minimum height=6mm,draw]
\tikzstyle{transvert}=[rectangle,minimum width=2mm, minimum height=8mm,draw]
\tikzstyle{transhor}=[rectangle,minimum width=8mm, minimum height=2mm,draw]
\tikzstyle{triangle}=[draw, regular polygon, regular polygon sides=3]
\tikzset{>={Stealth[scale=1.2]}}

\newcommand{\D}{\mathcal{D}}
\newcommand{\E}{\mathfrak{E}}
\newcommand{\F}{\mathcal{F}}

\newcommand{\M}{\mathcal{M}}
\newcommand{\N}{\mathcal{N}}

\newcommand{\T}{\mathcal{T}}

\newcommand{\X}{\mathcal{X}}

\newcommand{\fmset}[1]{\{\{#1\}\}}
\newcommand{\nestTok}{\T}

\usepackage{mathabx}

\usepackage{stmaryrd}

\newcommand*{\defeq}{\stackrel{\text{def}}{=}}
\newcommand{\nuPN}{$\nu$PN\xspace}
\newcommand{\var}{\text{Var}}

\begin{document}

\title{A Lower Bound on Conservative Elementary Object Systems Coverability}

\author{Francesco {Di Cosmo}\inst{1}\orcidID{0000-0002-5692-5681} \and
Soumodev Mal \inst{2} \orcidID{0000-0001-5054-5664} \and
Tephilla Prince \inst{3} \orcidID{0000-0002-1045-3033}
}

\authorrunning{{Di Cosmo} et al.}

\institute{Free University of Bozen-Bolzano, Italy \email{frdicosmo@unibz.it} \and
Chennai Mathematical Institute, India \email{soumodevmal@cmi.ac.in}\and
IIT Dharwad, India \email{tephilla.prince.18@iitdh.ac.in}
}

\maketitle

\begin{abstract}

Elementary Object Systems (EOS) are a form of Petri Net (PN) where tokens carry internal PN. 
This model has been recently proposed for analysis of robustness of Multi Agent Systems. 
While EOS reachability is known to be undecidable, the decidability of coverability of its conservative fragment (where the type of internal PN cannot be completely deleted and, thus, is conserved) was proved a decade ago, no study charted its complexity. 
Here, we take a first step in this direction, by showing how to encode $\nu$PNs, a well studied form of PN enriched with data, into conservative EOS (cEOS). 
This yields a non-Primitive Recursive, $F_{\omega2}$ lower-bound on cEOS coverability.

\end{abstract}

\section{Introduction}

Elementary Object Systems (EOS)~\cite{kanade_object_2004} are an extension of standard Petri Nets (PNs), where tokens carry internal PNs. Different variants of this model have been studied for the last decades (see, for example,~\cite{kohler-busmeier_survey_2014,lomazova00}). Recently, in~\cite{Bussmeier23a}, inspired by the \textit{let it crash} programming approach, EOSs have been proposed as a model to analyze the robustness of Multi Agent Systems suffering of full breakdowns of agents. Specifically, each token represents an agent; when its internal PN enters in a deadlock, the agent breaks. 
\cite{OurPNSE24} studied partial forms of breakdowns, reminiscent of channel imperfections in message passing systems, where the internal configuration of the tokens may non-deterministically degrade. The analysis of reachability and coverability problems for such imperfect EOSs generalizes the coverability problem for a notable fragment of perfect EOSs, namely conservative EOSs (cEOSs). Thus, the complexity of cEOS coverability immediately yields a lower bound for the imperfect counterparts.
A decade ago,~\cite{kohler-busmeier_survey_2014} showed that cEOS coverability is decidable. Yet, the complexity of cEOS coverability has remained an open problem. Here, we take a first step to obtain a lower bound. Technically, we encode \nuPN into cEOSs, which allows us to use the complexity of \nuPN coverability, which is a non-Primitive Recursive $\F_{\omega2}$ problem, as a lower-bound for cEOS coverability.
Preliminaries, including EOSs and \nuPN, are discussed up to Sec.~\ref{sec:nupn}, while the lowerbound is obtained in Sec.~\ref{Sec:lowerbound}. Sec.~\ref{sec:conclusions} draw the conclusions.

\section{Binary Relations and Multisets}\label{sec:prelims}
Let us fix some notation for binary relations. Given a set $X$, the \emph{identity relation} $\id_X$ on $X$ is the relation $\{(x,x)\in X^2\mid x\in X\}$. Given a binary relation $<$ on $X$, we denote its \emph{reflexive closure} $<\cup\id_X$ by $\leq$ and its \emph{anti-reflexive part}  $<\setminus \id_X$ by $\lneq$. We use the symbol $>$ to denote the relation such that, $x>y$ iff $y<x$. For example, if $<$ is transitive, then $\leq$ and $\geq$ are transitive and reflexive (i.e., quasi orders). The same applies to the symbols $\prec$ and $\succ$ and their closures. From now on, we use $\prec$ to denote arbitrary transitive relations, and $<$ (possibly with a subscript) to represent fixed transitive relations, e.g., the standard order of $\mathbb{N}$.

A \emph{multiset} $\multiset{m}$ on a set $D$ is a mapping $\multiset{m}:D\rightarrow \mathbb{N}$. The \emph{support} of $\multiset{m}$ is the set $\support{m} = \{i \mid \multiset{m}(i) > 0\}$. The multiset $\multiset{m}$ is finite if its $\support{\multiset{m}}$ is finite. The family of all multisets over $D$ is denoted by $D^\oplus$. We denote a finite multiset $\multiset{m}$ by enumerating the elements $d\in\support{\multiset{m}}$ exactly $\multiset{m}(d)$ times in between $\{\{$ and $\}\}$, where the ordering is irrelevant. For example, the finite multiset ${\multiset{m}:\{p,q\}\longrightarrow \mathbb{N}}$ such that $\multiset{m}(p)=1$ and $\multiset{m}(q)=2$ is denoted by $\fmset{p,q,q}$.
The empty multiset $\fmset{}$ (with empty support) is also denoted by $\emptyset$. On the empty domain $D=\emptyset$ the only defined multiset is $\emptyset$; to stress this out we denote this special case, i.e., the empty multiset over the empty domain, by $\varepsilon$.
Given two multisets $\multiset{m_1}$ and $\multiset{m_2}$ on $D$, we define $\multiset{m_1} + \multiset{m_2}$ and $\multiset{m_1} - \multiset{m_2}$ on $D$ as follows:
$(\multiset{m_1} + \multiset{m_2})(d) = \multiset{m_1}(d)  + \multiset{m_2}(d)$ and
$(\multiset{m_1} - \multiset{m_2})(d) = max(\multiset{m_1}(d)  - \multiset{m_2}(d),0)$.
Similarly, for a finite set $I$ of indices, $\sum_{i\in I} \fmset{d_i}$ denotes the multiset $\multiset{m}$ over $\bigcup_{i\in I}\{d_i\}$ such that $\multiset{m}(d)=|\{i\in I \mid d_i=d\}|$ for each $d\in D$. 
Where useful, with a slight abuse of notation, we omit the double brackets, i.e., $\sum_{i\in I} \fmset{d_i}= \sum_{i\in I} {d_i}$. If $I=\{1,\dots,n\}$, then $\sum_{i\in I} {d_i}=\sum_{i=1}^n {d_i}$.
Finally, we write $\multiset{m_1} \sqleq \multiset{m_2}$ if, for each $d\in D$, we have $\multiset{m_1}(d)  \leq \multiset{m_2}(d)$. 

\section{Petri Nets}
In this section, we fix the notation of standard PNs~\cite{murata89}. We denote a PN $N$ as a tuple $N=(P,T,F)$, where $P$ is a finite place set, $T$ is a finite transition set, and $F$ is a flow function. Where useful, we equivalently interpret $F$ via the functions ${\prefun}_N : T\rightarrow ( P\rightarrow \mathbb{N})$ where $\prefun_N(t)(p)=F(p,t)$ and ${\postfun}_N : T\rightarrow ( P\rightarrow \mathbb{N})$ where $\postfun_N(t)(p)=F(t,p)$. 
For example, a transition $t\in T$ is enabled on a marking $\mu$ (finite multiset of places) if, for each place $p\in P$, we have $\prefun_N(t)(p)\leq \mu(p)$. Its firing results in the marking $\mu'$ such that $\mu'(p)=\mu(p)-\prefun_N(t)(p)+\postfun_N(t)(p)$, for each $p\in P$. We denote markings according to multiset notation. For example, the marking $\mu$ that places one token on place $p$ and two on place $q$ is denoted by $\fmset{p,q,q}$. The empty marking is denoted by $\emptyset$.
We also work with the special \emph{empty PN} $\blacksquare=(\emptyset,\emptyset,\emptyset)$, whose only marking is $\varepsilon$.

\section{Elementary Object Systems}\label{sec:eos}
Here we give only the syntax and semantics of EOSs. The interested reader can find more details, motivations, and examples, in~\cite{OurPNSE24}.

\begin{definition}[\textbf{EOS}]\label{def:bussy14_eos}
An \emph{EOS} $\os$ is a tuple $\os=\tup{\hat{N},\N,d,\Theta}$ where:
\begin{compactenum}
\item $\hat{N}=\tup{\hat{P},\hat{T},\hat{F}}$ is a PN called \emph{system net}; $\hat{T}$ contains a special set $ID_{\hat{P}}=\{id_p\mid p\in \hat{P}\}\subseteq \hat{T}$ of \emph{idle transitions} such that, for each distinct $p,q\in \hat{P}$, we have $\hat{F}(p,id_p)=\hat{F}(id_p,p)=1$ and $\hat{F}(q,id_p)=\hat{F}(id_p,q)=0$.
\item $\N$ is a finite set of PNs, called \emph{object PNs}, such that $\blacksquare\in\N$ and if $(P_1,T_1,F_1), (P_2,T_2,F_2)\in\N \cup \hat{N}$,\footnote{This way, the system net and the object nets are pairwise distinct.} then $P_1\cap P_2=\emptyset$ and $T_1 \cap T_2 = \emptyset$.
\item $d:\hat{P}\rightarrow \N$ is called the \emph{typing function}. 
\item $\Theta$ is a finite \emph{set of events} where each \emph{event} is a pair $(\hat{\tau},\theta)$, where $\hat{\tau}\in \hat{T}$ and  $\theta:\N \rightarrow \bigcup_{(P,T,F)\in\N} T^\oplus$,
    such that $\theta((P,T,F))\in T^\oplus$ for each $(P,T,F)\in\N$ and, if $\hat{\tau}=id_p$, then $\theta(d(p)) \neq \emptyset$.
\end{compactenum}
\end{definition}
\begin{definition}[Nested Markings]
Let $\os=\tup{\hat{N},\N,d,\Theta}$ be an EOS. The set of \emph{nested tokens} $\nestTok(\os)$ of $\os$ is the set $\bigcup_{(P,T,F)\in\N} (d^{-1}{(P,T,F)}\times P^{\oplus})$. The set of \emph{nested markings} $\M(\E)$ of $\os$ is $\nestTok(\os)^{\oplus}$.
Given $\lambda,\rho\in \M(\E)$, we say that $\lambda$ is a \emph{sub-marking} of $\mu$ if $\lambda \sqleq \mu$.
\end{definition}

EOSs inherit the graphical representation of PNs with the provision that we represent nested tokens via a dashed line from the system net place to an instance of the object net where the internal marking is represented in the standard PN way. However, if the nested token is $\tup{p,\varepsilon}$ for a system net place $p$ of type $\blacksquare$, we represent it with a black-token $\blacksquare$ on $p$. If a place $p$ hosts $n>2$ black-tokens, then we represent them by writing $n$ on $p$. Each event $\tup{\hat{\tau},\theta}$ is depicted by labeling $\hat{\tau}$ by $\tup{\theta}$ (possibly omitting double curly brackets). If there are several events involving $\hat{\tau}$, then $\hat{\tau}$ has several labels.

\begin{definition}[{Projection Operators}]
Let $\os$ be an EOS $\tup{\hat{N},\N,d,\Theta}$. The \emph{projection operators $\Pi^1$} maps each nested marking $\mu=\sum_{i\in I}\tup{\hat{p}_i,M_i}$ for $\E$ to the PN marking $\sum_{i\in I}\hat{p}_i$ for $\hat{N}$. Given an object net $N\in\N$, the \emph{projection operators $\Pi^2_N$} maps each nested marking $\mu=\sum_{i\in I}\tup{\hat{p}_i,M_i}$ for $\E$ to the PN marking $\sum_{j\in J} M_j$ for ${N}$ where $J=\{i\in I\mid d(\hat{p}_i)=N\}$.
\end{definition}

To define the enabledness condition, we need the following notation. We set $\prefun_{N}(\theta(N))=\sum_{i\in I}\prefun_N(t_i)$ where $(t_i)_{i\in I}$ is an enumeration of $\theta(N)$ counting multiplicities. We analogously set $\postfun_{N}(\theta(N))=\sum_{i\in I}\postfun_N(t_i)$.

\begin{definition}[{Enabledness Condition}]\label{def:bussy14_enable} 
Let $\os$ be an EOS $\tup{\hat{N},\N,d,\Theta}$. Given an event $e=\tup{\hat{\tau},\theta}\in \Theta$ and two markings $\lambda,\rho\in\M(\os)$, the \emph{enabledness condition} $\Phi(\tup{\hat{\tau},\theta},\lambda,\rho)$ holds iff
\begin{align*}
\Pi^1(\lambda)=\prefun_{\hat{N}}(\hat{\tau})\ \land \Pi^1(\rho)=\postfun_{\hat{N}}(\hat{\tau})\ \land
\forall N\in \N,\ \Pi^2_N(\lambda)\geq \prefun_N(\theta(N))\ \land\\
\forall N\in\N,\ \Pi^2_N(\rho)=\Pi^2_N(\lambda)-\prefun_N(\theta(N))+\postfun_N(\theta(N))
\end{align*}
The event $e$ is \emph{enabled with mode $(\lambda,\rho)$ on a marking $\mu$} iff $\Phi(e,\lambda,\rho)$ holds and $\lambda\sqleq \mu$.
Its firing results in the step $\mu\xrightarrow{(e,\lambda,\rho)}\mu-\lambda+\rho$.
\end{definition}

The coverability problem for EOSs is defined in the usual way, i.e., it asks whether there is a run (sequence of event firings) from an initial marking $\mu_0$ to a marking $\mu_1$ that covers a target marking $\mu_f$ with respect to the order $\leq_f$ such that $\mu\leq_f\mu'$ iff $\mu'$ is obtained from $\mu$ by adding
\begin{inparaenum}[\itshape (1)]
    \item tokens in the inner markings of available nested tokens and/or
    \item nested tokens with some internal marking on the system net places.
\end{inparaenum}

It is known that EOS coverability is undecidable (Th. 4.3 in~\cite{kohler-busmeier_survey_2014}. However, coverability is decidable on the fragment of \emph{conservative EOSs} (cEOSs; Th. 5.2 in~\cite{kohler-busmeier_survey_2014}), where, for each system net transition $t$, if $t$ consumes a nested token on a place of type $N$, then it also produces at least one token on a place of the same type $N$. 

\begin{definition}[{cEOS}]
A cEOS is an EOS $\os=\tup{\hat{N},\N,d,\Theta}$ with $\hat{N} = \tup{\hat{P},\hat{T},\hat{F}}$ where, for all $\hat{t} \in \hat{T}$ and $\hat{p} \in \support{\prefun_{\hat{N}}(\hat{t})}$, there exists $\hat{p}' \in \support{\postfun_{\hat{N}}(\hat{t})}$ such that $d(\hat{p}) = d(\hat{p}')$.
\end{definition}

\section{\texorpdfstring{$\nu$PNs}{nuPNs}}\label{sec:nupn}
We introduce $\nu$PN as in~\cite{LazicS16}. Let $\Upsilon$ and $\X$ be disjoint infinite sets variables. The variables in $\X$ are called \textit{standard} variables, while those in $\Upsilon$ are called \textit{fresh}. Let $Vars\defeq \X \bigcup \Upsilon$.
\begin{definition}\label{dfn:nuPN}
A \nuPN  is a tuple $\D=\tup{P,T,F}$ where
\begin{inparaenum}[\itshape (1)]
    \item $P$ is a finite non-empty set of places,
    \item $T$ is a finite set of transitions disjoint from $P$,
    \item $F:(P\times T) \bigcup (T \times P) \to Vars^{\oplus}$ is a flow function such that, for each $t\in T$, $\Upsilon\cap\prefun(t)=\emptyset$ and $\postfun(t)\setminus\Upsilon\subseteq\prefun(t)$, where $\prefun(t)=\bigcup_{p\in P} \supp(F(p,t))$ and $\postfun(t)=\bigcup_{p\in P} \supp(F(t,p))$.
\end{inparaenum}
\end{definition}

For each $t\in T$, we set $\var(t)=\prefun(t)\cup\postfun(t)$. In this section, we work with a fixed arbitrary \nuPN $\D=\tup{P,T,F}$ where $P=\{p_1,\dots,p_\ell\}$.
The flow $F_x$ of a variable $x\in \var$ is $F_x:(P\times T) \bigcup (T \times P) \to \mathbb{N}$ where $F_x (p,t)\defeq F(p,t)(x)$ and $F_x(t,p)\defeq F(t,p)(x)$. We denote the vector $\tup{F(p_1,t),\dots,F(p_\ell,t)}\in\mathbb{N}^\ell$ by $F_x(P,t)$ and the vector $\tup{F(t,p_1),\dots,F(t,p_\ell)}\in\mathbb{N}^\ell$ by $F_x(t,P)$.

The set of configurations of $\D$ is the set $(\mathbb{N}^P)^\oplus$. To define the transitions, we introduce:
\[
    in(t)\defeq \sum_{x\in \X(t)} \fmset{F_x(P,t)}
\quad
    \out\defeq \sum_{\nu\in \Upsilon(t)} \fmset{F_\nu(t,P)}
\]
Given a configuration $M=\fmset{m_1,\dots,m_{\abs{M}}}$, a transition $t$ is fireable from $M$ if there is a function $e:\X(t)\longrightarrow\{1,\dots,\abs{M}\}$, called mode, such that, for each $x\in\X(t)$, $F_x(P,t)\leq m_{e(x)}$. We write $M\rightarrow^{t,e} M'$ if 

\[ M=M''+\sum_{x\in\X(t)}\fmset{m_{e(x)}}
\quad
M'=M''+\out+\sum_{x\in\X(t)}\fmset{m_{e(x)}'}
\]
where, for each $x\in\X(t)$, we have $m'_{e(x)}=m_{e(x)}-F_x(P,t)+F_x(t,P)$. Intuitively, the firing of $t$ with mode $e$ over $M$ applies $F_x$, for each $x\in X(t)$ to a distinct tuple $m\in M$ such that $m_{e(x)}\geq F_x(P,t)$ and replaces it with $m'_{e(x)}$. Moreover, it adds new markings $F_{\nu}(t,P)$ for each $\nu\in\Upsilon$.

Without loss of generality, we assume that the net $\D$ is normal~\cite{Rosa-Velardo11}, that is, there is a special variable $\nu\in\Upsilon$ such that, for each $t\in T$ and $p\in P$, either $F(t,p)\sqsubseteq \X^\oplus$ or $F(t,p)=\fmset{\nu}$.

The \textit{size} $\abs{\D}$ of a $\D$ is measured in terms of a unary encoding of the coefficients in the multisets defined by the flow function. Specifically, $\abs{\D}=\max(\abs{P},\abs{T},\sum_{p\in P, t\in T}\abs{F(p,t)}+\abs{F(t,p)})$.
The \textit{coverability} problem for \nuPN asks to check, for a given \nuPN, a target configuration $\tau$, initial configuration $\iota$, if there is some configuration $\tau'$ reachable from $\iota$ such that $\tau\sqsubseteq \tau'$. It is known that \nuPN coverability is $\mathbf{F}_{\omega 2}$-complete~\cite{LazicS16}, i.e., double-Ackermann complete. 

\section{Complexity Lower Bound}\label{Sec:lowerbound}

We now show how to encode $\nu$PNs into cEOS. The idea of our reduction is that we can encode each tuple of a configuration of a \nuPN $\D$ into a dedicated object net inside a \textit{simulator} place at the system net level.
 
We now define an EOS $\os$ that can be used to check coverability in $\D$.  
 We start by defining the object nets. Actually, $\os$ involves only the type $\blacksquare$, used to fire sequences of events, and one more type $N_\D$ that captures the transitions of $\D$ split by variable as specified in the next definitions and example. Since we deal with only two types, in this section we graphically represent system net places of type $N_\D$ by circles, as usual, while we represent places of type $\blacksquare$ by triangles.

$N_\D$ includes all places of $\D$ and the restriction of each transition of $\D$ to each of its variables.

\begin{definition}
    $N_\D$ is the PN $N_\D=(P_\D,T_\D,F_\D)$ such that $P_\D=P$, $T_\D=\biguplus_{t\in T}\{t_x\mid x\in\var(t)\}$, and, for each $p\in P_\D$ and $t_x\in T_\D$, $F_\D(p,t_x)=F_x(p,t)$ and $F_\D(t_x,p)=F_x(t,p)$.
\end{definition}

The system net $\hat{N}=(\hat{P},\hat{T},\hat{F})$ of $\os$ contains the places $\text{sim}$, $\text{selectTran}$, and, for each $t\in T$, the places, transitions, and flow function depicted in Fig.~\ref{fig:lowerboundnuPNtoEOSblock}. The figure depicts also the synchronization structure $\Theta$ of $\os$. Note that the only transitions that can fire concurrently are the $t^\text{fire}_{x}$, for $x\in\var(t)$.\begin{figure}[t]\centering
\resizebox{1\columnwidth}{!}{
\begin{tikzpicture}

\node[triangle,label={[name=selLab]above:\scriptsize $\text{selectTran}$}](sel)at(0,0){};
\node[transhor,label={[name=t1Lab]above:\scriptsize $t_{x_1}^{\text{select}}$}](t1)at(1.5,0){};
\node[triangle,label={[name=sel1Lab]above:\scriptsize $\text{select}_{x_2}^t$}](sel1)at(3,0){};
\node at(4.5,0)(selDots){$\dots$};
\node[triangle,label={[name=t2Lab]above:\scriptsize $\text{select}_{x_n}^t$}](t2)at(6,0){};
\node[transhor,label={[name=sel2Lab]above:\scriptsize $t_{x_n}^\text{select}$}](sel2)at(7.5,0){};

\node[triangle,label={[name=tnuLab]above:\scriptsize $\text{select}_{\nu}^t$}](tnu)at(9,0){};
\node[transhor,label={[name=selnuLab]above:\scriptsize $t_\nu^\text{select}$}](selnu)at(10.5,0){};

\node[place,label={[name=simLab]left:\scriptsize $\text{sim}$}](sim)at(-1.5,1.25){};

\node[place,label={[name=t1selLab]left:\scriptsize $t^\text{selected}_{x_1}$}](t1sel)at(1.5,-1){};
\node[triangle,label={[name=t1runLab]right:\scriptsize $t^\text{run}_{x_1}$}](t1run)at(2.5,-1){};

\node[transhor,label={[name=t1putLab]right:\scriptsize $t^\text{fire}_{x_1}\tup{t_{x_1}}$}] at (2,-2)(t1put){};

\node[place,label={[name=tnselLab]left:\scriptsize $t^\text{selected}_{x_n}$}](tnsel)at(7.5,-1){};
\node[triangle,label={[name=tnrunLab]right:\scriptsize $t^\text{run}_{x_n}$}](tnrun)at(8.5,-1){};

\node[transhor,label={[name=tnputLab]left:\scriptsize $t^\text{fire}_{x_n}\tup{t_{x_n}}$}] at (8,-2)(tnput){};

\node[triangle,label={[name=tnurunLab]left:\scriptsize $t^\text{run}_{\nu}$}](tnurun)at(10.5,-1){};

\node[transhor,label={[name=tnuputLab]left:\scriptsize $t^\text{fire}_{\nu}\tup{t_\nu}$}] (tnuput)at (10.5,-2){};

\draw[->] (sim.south east) -| (t1Lab.north);
\draw[->] (sim.east) -| (sel2Lab.north);

\draw[->] (t1.south) -- (t1sel.north);
\draw[->] (sel2.south) -- (tnsel.north);

\draw[->] ($(selnu.south)-(.2,0)$) |- ($(t1run)+(0,.75)$) -- (t1run);
\draw[->] ($(selnu.south)-(.1,0)$)|- ($(tnrun)+(0,.6)$) -- (tnrun);
\draw[->] ($(selnu.south)-(0,0)$) -- (tnurun);

\draw[->] (sel) -- (t1);
\draw[->] (t1) -- (sel1);
\draw[->] (sel1) -- (selDots);
\draw[->] (selDots)--(t2);
\draw[->] (t2)--(sel2);
\draw[->] (sel2)--(tnu);
\draw[->] (tnu)--(selnu);

\draw[->] (t1put) -| ($(sim.south)+(.2,0)$);
\draw[->] ($(tnput.south) - (.1,0)$) -- ($(tnput.south)-(.1,.15)$) -| (sim);
\draw[->] ($(tnuput.south) - (.1,0)$)-- ($(tnuput.south)-(.1,.3)$)-| ($(sim.south)-(.2,0)$);

    \node[triangle,label={[name=reportLab]above right:\scriptsize $t^\text{report}$}](report)at (2,-3.1){};

\node[transhor,label={[name=doneLab]left:\scriptsize $t^\text{done}$}](done)at(0,-3.1){};

\draw[->] (t1put) -- (report);
\draw[->] (tnput) |- (report.north east);
\draw[->] (tnuput) |- (report.east);
\draw[->] (report) --node [midway, above]{\scriptsize 
$n+1$} (done);
\draw[->] (done) -- (sel);

\node at (sel){\scriptsize $\blacksquare$};

\draw[->](t1sel) -- (t1put);
\draw[->](t1run) -- (t1put);

\draw[->](tnsel) -- (tnput);
\draw[->](tnrun) -- (tnput);

\draw[->](tnurun) -- (tnuput);

\end{tikzpicture}
}

    \caption{The part of $\os$ dedicated to the simulation of a transition $t$ of $\D$ such that $\var(t)=\{x_1,\dots,x_n,\nu\}$. If $\nu\notin\var(t)$, then, $\text{select}^t_\nu$, $t^\text{select}_\nu$, $t_\nu^\text{run}$, and $t_{\nu}^\text{fire}$ have to be dropped, and $\hat{F}(t^\text{select}_x,t_{x_1}^\text{run})=1$, $\hat{F}(t^\text{select}_x,t_{x_n}^\text{run})=1$, and $\hat{F}(t^\text{report},t^\text{done})=n$ has to be set.
    }
    \label{fig:lowerboundnuPNtoEOSblock}
\end{figure}
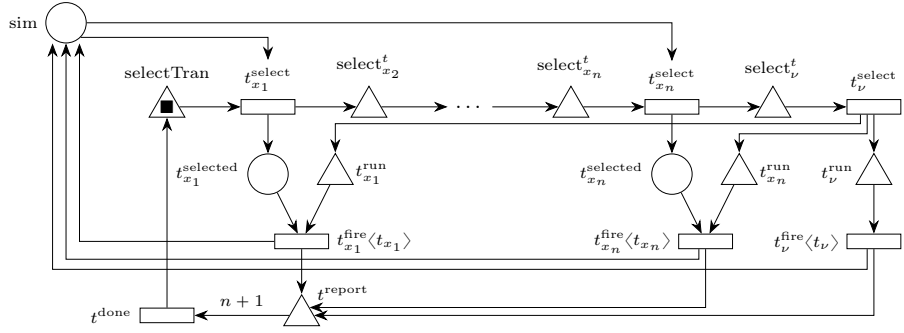

Intuitively, given a marking $\mu$ of $\D$, we capture $\mu$ by using, for each tuple $m\in\mu$ a dedicated $N_\D$ object in a place \textit{sim} with internal marking $m$. 
\begin{definition}\label{def:lowerBoundEncode}
    The configuration $M=\fmset{m_1,\dots,m_\ell}$ of $\D$ is encoded by the configuration $\overline{M}=\sum_{i=1}^\ell \tup{\text{sim}, m_i} + \tup{\text{selectTran}, \varepsilon}$.
\end{definition}

To simulate a transition $t$ of $\D$, we just need to, for each $x\in\var(t)\setminus\{\nu\}$,
\begin{inparaenum}[\itshape (1)]
    \item map $x$ to some object $O_x$ in \textit{sim} and then
    \item concurrently fire the transitions $t_x$ in each $O_x$.
\end{inparaenum} $\os$ performs this mapping by moving the non-deterministically chosen object $O_x$ from $\text{sim}$ to the place $t^{\text{selected}}_{x}$, for $x\in\var(t)\setminus\{\nu\}$. This selection happens sequentially as enforced by the array of places $\text{selectTran}$, $\text{select}^t_{x_i}$, and transitions $t^\text{select}_{x_i}$ and, if present, $\text{select}^t_{\nu}$ and $t^\text{select}_{\nu}$.

Only after firing the last $t^\text{select}_y$ transition,\footnote{That is $y=\nu$ if $\nu\in\var(t)$ and $y=x_n$ otherwise.} the transitions $t^\text{fire}_{x}$ get enabled and, since they are synchronized with $t_x$, they \begin{inparaenum}[\itshape (1)]
    \item move (or create, if $x=\nu$) the object $O_x$ from $t^\text{selected}_{x}$ back to $\text{sim}$,
    \item update the internal marking of $O_x$ according to $t_x$, and
    \item move a $\blacksquare$ token from $t^{\text{run}}_x$ to $t^\text{report}$.
\end{inparaenum}

After firing all the $t^\text{fire}_{x}$ transitions, the $t^\text{done}$ transition gets enabled. Its firing cleans $t^\text{report}$ and produces a single $\blacksquare$ token in $\text{sim}$, completing the simulation of $t$.

The next lemma follows from the previous description. Given some configuration $M$ and $M'$ of $\D$, we say that a run $\overline{M}\rightarrow^*\overline{M'}$ of $\os$ of positive length is \textit{minimal} if in each intermediate configuration the place $\text{selectTran}$ is not marked.
\begin{lemma}
    $M\rightarrow^t M'$ in $\D$ if only if $\overline{M}\rightarrow^{\sigma} \overline{M'}$ in $\os$ for some minimal run $\sigma$ of $\os$.
\end{lemma}
\begin{proof}
We assume $\var(t)=\{x_1,\dots,x_n,\nu\}$: in case $\nu\notin\var(t)$, the argument is analogous. 

If $M\rightarrow^t M'$, then there is a mode $e$ such that $M\rightarrow^{t,e}M'$. Thus, we can write $M=\fmset{m_1,\dots,m_n}+M''$ for some $m_1,\dots,m_n\in\mathbb{N}^\ell$ and configuration $M''$ of $\D$. By \nuPN semantics, $M'=M''+F_\nu(P,t)+\sum_{i=1}^n m_i - F_{x_i}(P,t) + F_{x_i}(t,P)$.
Thus, by Definition~\ref{def:lowerBoundEncode}, $\overline{M}=\overline{M''}+\sum_{i=1}^n \tup{\text{sim},m_i}$ and $\overline{M'}=\overline{M''}+\tup{\text{sim},F_\nu(P,t)}+\sum_{i=1}^n\tup{\text{sim},m_i - F_{x_i}(P,t) + F_{x_i}(t,P)}$.

The transition $t^{\text{select}}_{x_1}$ is enabled on $\overline{M}$. Moreover, we can fire the sequence of transitions $t^{\text{select}}_{x_1},\dots,t^{\text{select}}_{x_n}$ so as to select objects in $\text{sim}$ matching the mode $e$, i.e., by reaching a configuration $M''-\tup{\text{selectTran},\varepsilon}+\sum_{i=1}^n\tup{t^{\text{selected}}_{x_i},m_{e(x_{i})}}$. 

After firing $t^{\text{select}}_\nu$, for $x\in\var(t)$, the transitions $t^{\text{fire}}_{x}$ are enabled and, because of the synchronization structure, their firing results in the configuration
$M''-\tup{\text{selectTran},\varepsilon}+\sum_{i=1}^n\tup{\text{sim},m_{e(x_{i})} 
-F_{x_i}(P,t)+ F_{x_i}(t,P)} + \tup{\text{sim},F_{\nu}(t,P)} + (n+1)\tup{t^{\text{report}},\varepsilon}$. On this configuration, $t^{\text{done}}$ is enabled and its firing returns the configuration $\overline{M'}$. Thus, we have exhibited a run $\sigma$ such that $\overline{M}\rightarrow^{\sigma}\overline{M'}$. Moreover, this sequence is minimal.

Vice-versa, if there is a minimal run $\sigma$ of $\os$ such that $\overline{M}\rightarrow^\sigma\overline{M''}$, then we can organize $\sigma$ in three blocks: 
\begin{inparaenum}[\itshape (1)]
\item a prefix $\sigma'$ that amounts to the firing of the $t^\text{select}_x$ transitions,
\item an intermediate run $\sigma''$ that amounts to the firing of the transitions $t^{\text{fire}}_{x}$, and
\item a suffix $\sigma'''$ that amounts to the firing of $t^{\text{report}}$, for $x\in\var(t)$.
\end{inparaenum}

Moreover, the configuration reached at the end of $\sigma'$ is $\overline{M''}-\tup{\text{selectTran},\varepsilon} + \sum_{i=1}^n \tup{t^{\text{selected}}_{x_i},m_i} + \sum_{x\in\var(t)}\tup{t^\text{run}_{x_i},\varepsilon}$. Hence, $\overline{M}=\overline{M''} +\sum_{i=1}^n\tup{\text{sim},m_i}$. Moreover, by applying $\sigma$ to $\overline{M}$, we obtain that $\overline{M'}=\overline{M''}+\tup{\text{sim},F_\nu(t,P)}+\sum_{i=1}^n \tup{\text{sim}, m_i - F_{x_i}(P,t) + F_{x_i}(t,P)}$.

Consequently, $M=M''+\sum_{i=1}^n m_i $ and $M'=M''+ F_\nu(t,P) + \sum_{1=1}^n m_i - F_{x_i}(P,t)+F_{x_i}(t,P)$.
Thus, $t$ is enabled on $M$ with the mode $e$ such that $e(x_i)=m_i$, for each $i\in\{1,\dots,n\}$ and $M\rightarrow^{t,e}M'$.

\end{proof}

Consequently, the set of reachable configurations in $\os$ that encode some configuration in $\D$ essentially matches the set of configurations reachable in $\D$. Thus, we can check coverability in $\D$ of a configuration $\tau$ from an initial configuration $\iota$ by checking coverability in $\os$ of $\overline{\tau}$ from $\overline{\iota}$. Thus, we have a reduction from \nuPN coverability to cEOS coverability. Note that this reduction is polynomial, since the construction of $\os$ is polynomial. Since \nuPN coverability is $\mathbf{F}_{\omega2}$-complete we obtain the following theorem.

\begin{theorem}
    cEOS coverability is $\mathbf{F}_{\omega2}$-hard.
\end{theorem}

\section{Conclusions}\label{sec:conclusions}

We have taken a first step towards the study of cEOS coverability complexity. Specifically, we have showed that cEOS coverability is non-Primitive Recursive and $\F_{\omega2}$-hard since it enjoys a polynomial reduction from \nuPN coverability. This raises questions about the precise relationship of cEOSs with \nuPN and, more in general, with data nets.

%\pagebreak
\bibliographystyle{splncs04}
\bibliography{references}
\end{document}